\documentclass{article} 

\usepackage{mathrsfs}     
\usepackage{bussproofs}	  
\usepackage{tcolorbox}

\usepackage{hyperref}
\usepackage{graphicx}
\usepackage{amsmath}
\usepackage{amssymb}
\usepackage{amstext}
\usepackage{amsthm}
\usepackage{mathtools}
\usepackage{comment}
\usepackage{todonotes}

\usepackage{color}
\usepackage{tikz}
\usetikzlibrary{positioning,arrows,calc,shapes,babel,fit}  
\tikzset{
modal/.style={>=stealth',shorten >=1pt,shorten <=1pt,auto,node distance=1.5cm,
semithick},
world/.style={circle,draw,minimum size=0.5cm,fill=gray!15},
tworld/.style={circle,draw,white,minimum size=0.7cm},
arg/.style={circle,draw,minimum size=0.5cm},
point/.style={circle,draw,inner sep=0.5mm,fill=black},
reflexive above/.style={->,loop,looseness=7,in=120,out=60},
reflexive below/.style={->,loop,looseness=7,in=240,out=300},
reflexive left/.style={->,loop,looseness=4,in=150,out=210},
reflexive right/.style={->,loop,looseness=4,in=30,out=330}
}

\newcommand{\black}{\color{black}}

\newcommand{\hf}[2]{\hyperref[#1]{#2}}

\newcommand{\depth}{\mathsf{depth}}

\newcommand{\atm}{\mathsf{At}}

\newcommand{\bi}{ {\langle\, b\,\rangle}} 
\newcommand{\nbi}{ {[\, b\,]}}

\newcommand{\lanset}{\mathcal{L}}
\newcommand{\lb}{{\cal L}_{{\square [\,b\,]}}}

\newcommand{\basiclan}{\lanset_\square}

\newcommand{\fmod}{\mathfrak{M}}
\newcommand{\ffra}{\mathfrak{F}}
\newcommand{\cmod}{\mathcal{M}}
\newcommand{\cfra}{\mathcal{F}}

\newcommand{\fset}{ {\textsf{Fclass}} }

 \newtheorem{ex}{Example}
  \newtheorem{definition}{Definition}
  \newtheorem{theorem}{Theorem}
  \newtheorem{lemma}{Lemma}
  \newtheorem{proposition}{Proposition}
   \newtheorem{remark}{Remark}
  \newtheorem{convention}{Convention}
  \newtheorem{corollary}{Corollay}

 \numberwithin{equation}{section}

\begin{document}

\title{A meta-modal logic for bisimulations}
\author{Alfredo Burrieza \\ {\small burrieza@uma.es}\\ {\small Universidad de M\'alaga} \\  \\ Fernando Soler-Toscano \\ {\small fsoler@us.es} \\ {\small Universidad de Sevilla} \\ \\ Antonio Yuste-Ginel \\ {\small ayusteginel@uma.es} \\{\small Universidad de M\'alaga}}

\maketitle

	\begin{abstract}
		    We propose a novel modal study of bisimulations. Our contribution is threefold. First, we extend the basic modal language with a new modality $\nbi$, whose intended meaning is universal quantification over all states that are bisimilar to the current one under a given bisimulation. We show that bisimulations are definable in this object language via frame validity (for the forth and back conditions) and model validity (for the atomic harmony condition). Additionally, some other related notions (refinements, simulations, bisimulations for a restricted set of variables, and $n$-bisimulations) are also shown to be definable. Second, we provide a sound and complete axiomatisation of the class of all pairs of Kripke models that are bisimulation-related. Third, we show that the satisfiability problem of this logic is decidable and PSPACE-complete via a translation to standard modal logic $K$ under a simple frame condition. All our results are encoded and verified by Isabelle/HOL.
	\end{abstract}

\sloppy

\section{Introduction}\footnote{The authors thank Quentin Gougeon for pointing us to some relevant literature on the topic of this paper. They are also grateful to anonymous reviewers of the \textit{Review of Symbolic Logic} and \textit{Advances in Modal Logic} for commenting on previous versions of this work.}

Modal logic can be understood in a narrow, classical sense as the logical study of alethic modalities: necessity, possibility, contingency and impossibility. In a broader sense, (propositional) modal logic is an extension of propositional logic with modal operators, which are operators that take other formulas as arguments. This also encompasses the study of how these extended languages are useful for describing certain mathematical structures \cite{blackburn2010modal} (see also \cite{chellas1980modal,hughes1996new,garson2013modal,van2010modal,zalta1995basic}). \par 

In the latter, broader sense, there is a long-standing tradition in which modal operators are interpreted over a chosen mathematical concept (this can be contrasted with another tradition promoting an intuitive or philosophical reading of modal operators, such as knowledge/belief \cite{hintikka1962knowledge}, time \cite{prior1957time}, or moral obligation \cite{von1951deontic}.)
Then, such a concept is studied in a modal disguise, and its axiomatisation provides researchers with an informative, fresh look at the idea. Moreover, modalization of mathematical theories enables the importation of modal-logic tools beyond proof methods, including model-theoretic and computational machinery. This tradition can be dated back at least to the work of Tarski and McKinsey on the topological interpretation of the basic modal language \cite{mckinsey1944algebra}, which demonstrated that a rather intricate mathematical structure (topologies) admits a relatively simple and well-behaved normal modal logic (S4). \par 

A further step in this research direction is to provide a modal approach to notions that belong to the meta-mathematical theory of modal logic. This naturally results in what might be called \textit{meta-modal logic}: the construction of object modal languages to study concepts stemming from classical modal meta-theory. 
In this vein, much recent work has focused on the modal approach of filtrations \cite{baltag2017quotient,ilin2018filtration,vansome}, a well-known method for obtaining the finite model property of modal logics and therefore a direct path to decidability. Here, we propose to focus on another modal logic concept: bisimulations. \par 

Bisimulations are a fundamental formal tool in the model theory of standard modal logic (see \cite[Chapter 2]{blackburn2010modal} and \cite{sangiorgi2009origins} for a systematic presentation as well as historical notes; Johan van Benthem was the pioneer of the notion, under the name of p-relations \cite{benthen1976thesis}.) Roughly speaking, bisimulations gave a clear answer to a foundational model-theoretical question: Given two (Kripke-style) models, can we find relatively weak, language-independent conditions under which they satisfy the exact modal formulas? Bisimulations answer this question positively.
Among many other applications, bisimulations are used to yield definability and undefinability results, as well as tight correspondence results between first-order and modal logic. Additionally, they provide techniques for reducing model size, which in turn has a direct practical payoff.\par 
\black

\par 

Previous works have already approached bisimulations from a modal perspective. In particular, \cite{hollenberg1998thesis,french2005bisimulation,french2006bisimulation,french2006bisimulation} introduced and studied so-called \textit{bisimulation quantified modal logics}, where formulas $\forall p A$ are true in a given state $s$ if and only if $A$ is true in all $\{p\}$-bisimilar states (states that are bisimilar as far as only atom $p$ is considered). In another line of work, \cite{van2014definability,fernandez2011modal}, among others, addressed the modal definability of several classes of models up to bisimulation via modal and local truth (see \cite{van2011logical} for a textbook presentation of the topic). We here propose a conceptual shift: studying bisimulations from the perspective of correspondence theory (the theory connecting modal and first-order logic), where the main background semantic notions are frame and model validity, instead of truth. \par 
Our contribution to this enterprise is threefold. First, we show that the main definitional conditions of bisimulations between two arbitrary Kripke models, usually called  ``atomic harmony'',``forth'' (or ``zig'') and ``back'' (or ``zag''), can be modally defined using a single normal modality (which we will denote by $\nbi$). Intuitively, $\nbi  A$ means that $ A$ holds at every state that is linked to the current state by a given bisimulation. Second, we provide a sound and complete axiomatisation for the class of all pairs of Kripke models that are linked by a bisimulation (``bi-models'', in our terminology). Completeness is proved via a canonical model technique, but it requires some creativity and deviation from the fully standard constructions. Third, we prove decidability and PSPACE-completeness of the satisfiability problem of the logic by translating it to the basic modal logic $K$ plus a simple frame axiom. Along the way, several notions related to bisimulations are also shown to be definable and axiomatisable in our setting: refinements, simulations, $\Phi$-bisimulations and $n$-bisimulations (the latter in a simple extension of our basic language). \par 

Our methodology is also double-sided. On the one hand, we propose a novel application of standard arguments and methods in the literature on modal definability, completeness and decidability/complexity. 
On the other hand, most of our central proofs have been fully encoded and verified~\cite{Bisimulation_Logic-AFP} in Isabelle/HOL~\cite{nipkow2002isabelle}, enhancing the reproducibility and applicability of our approach. Indeed, Isabelle identified some flaws and unjustified steps in our original handwritten proofs, which compelled us to reconsider several key points of the completeness argument.
\par 

{The rest of the paper is structured as follows. Section \ref{sec:preliminaries} introduces the needed background on basic modal logic. Section \ref{sec:logic} presents the modal language for bisimulations, its semantics, some model-theoretical results, and an axiomatic proof system. The latter is shown to be sound and strongly complete with respect to the class of all bi-models in Section \ref{sec:completeness}. Decidability and complexity of the satisfiability problem are discussed in Section \ref{sec:decidability}. Section \ref{sec:nbisimulations} is devoted to the extension of previous results to $n$-bisimulations (bisimulations up to modal depth $n$). Related work is discussed in Section \ref{sec:discussion}. We conclude with a brief recap, some discussion and open challenges in Section \ref{sec:conclusion}.}

\section{Background}\label{sec:preliminaries}

We assume a denumerable set of atoms $\atm$ as fixed from now on. 

\begin{definition}[Basic modal language]
The language ${\cal L}_{\square}$ is the one generated by the following grammar:
\begin{center}
$A::=  p \mid \perp \mid (A \to A) \mid \square A $
\end{center}
\noindent where $p$ ranges over $\atm$. 
\end{definition}

Atoms and their negations are called \textit{literals}.
We can introduce new defined operators. The rest of the Boolean connectives $\top,\lnot,\land,\lor,\leftrightarrow$ are defined as usual.
The dual of $\square$ is $\lozenge$, where $\lozenge A$ is shorthand for $\neg\square\neg A$. 
\par 
 The modal depth of a formula in the basic language is given by the function $\mathsf{depth(A)}:\basiclan\to \mathbb{N}$, define inductively as follows: $\depth(\perp)=\depth(p)=0$, $\depth(\lnot A)=\depth(A)$, $\depth(A\to B)=max\{\depth(A),\depth(B)\}$ and $\depth(\square A)=1+\depth(A)$.
 
 \begin{definition}[Basic frames, models and truth]\label{basica}
A {\em Kripke frame} for ${\cal L}_\square$ is an ordered pair of the form $\cfra=(W, R)$, where  $W\neq\varnothing$ and $R\subseteq W\times W$. 
We define $\fset = \{\cal F\mid \cal F \hbox{ is a Kripke frame}\}$. 
\par 
A {\em Kripke model} {${\cal M}$} is a tuple of the form $(W, R, V)$, where $(W, R)$ is a Kripke frame and $V\colon \atm\longrightarrow 2^W$ is a \textit{valuation function} that assigns a set of worlds to each atom; intuitively, the set of worlds where the atom is true. The model $(W, R, V)$ is said to be {\em based on} the frame $(W,R)$. {If $w\in W$ (that is, $w$ is in ${\cal M}$), then $({\cal M}, w)$ is called a pointed model. {We tend to omit brackets when denoting pointed models.}}
\par 
We also define the satisfaction relation $\models$, so that ${\cal M}, w\models A$ means `the formula $A$ is {\em satisfied at }(or {\em true at}) $w$  in the model ${\cal M}$'. Inductively:\par\medskip
 \begin{tabular}{l c l}
 ${\cal M}, w\models p$& iff & $w\in V(p)$ \quad (for any $p\in\atm$)\\[1.5mm]
   ${\cal M}, w\not\models \perp $ &\\[1.5mm]
  ${\cal M}, w\models A \to B$ & iff & either ${\cal M}, w\not\models A$ or ${\cal M}, w\models  B$ \\[1.5mm]
 ${\cal M}, w\models \square A$ & iff & for all $w'\in W$ s.t. $wRw'$,   ${\cal M}, w'\models   A$

\end{tabular}
\end{definition}

 The defined connective $\lozenge$ can be shown to have the following meaning: \par \medskip
 \begin{tabular}{l c l}
     ${\cal M}, w\models \lozenge A$ & iff &  there exists $w'\in W$ s.t. $wRw'$ and ${\cal M}, w'\models  A$\\[1.5mm]
\end{tabular}

\begin{definition}[Bisimulation for $\lanset_{\square}$]\label{def:bisimulation} Let ${\cal M} = (W, R, V)$ and  ${\cal M'} = (W', R', V')$ be a pair of Kripke models. We say that a non-empty binary relation $Z\subseteq W\times W'$ is  a {\em bisimulation between}  ${\cal M}$ and  ${\cal M'}$ (in symbols: $Z \cal \colon M\,\underline{\leftrightarrow} \,M'$) if the following conditions hold:
\begin{enumerate}
\item[(1)]  If $wZw'$, then $w$ and $w'$ satisfy the same atoms {\em (\textsf{atomic harmony})}.
\item[(2)]  If $wZw'$ and $wRv$, then there exists $v'$ in ${\cal M'}$ such that $vZv'$ and $w'R'v'$ \emph{(\textsf{forth})}.
\item[(3)]   If $wZw'$ and $w'R'v'$, then there exists $v$ in ${\cal M}$ such that $vZv'$ and $wRv$ (\emph{\textsf{back}}).
\end{enumerate}
If $w$ (in ${\cal M}$) and $w'$ (in ${\cal M'}$) are a pair of states linked by a bisimulation $Z$, then we say that $w$ and $w'$ are {\em bisimilar} (written $Z\colon {\cal M}, w\, \underline{\leftrightarrow}\, {\cal M'}, w'$ or if the context is clear, simply, $w\,\underline{\leftrightarrow}\, w'$). \par 
{We consider three additional notions obtained by weakening the definition of bisimulation. A \textit{$\Phi$-bisimulation} is a bisimulation where \textsf{atomic harmony} is only satisfied for $\Phi\subseteq \atm$. A \textit{simulation} is a bisimulation without the \textsf{back} condition. A refinement is a bisimulation without the \textsf{forth} condition.}

\end{definition}

\black 
The following well-known theorem in modal logic highlights the fundamental importance of bisimulations as a tool for logical inquiry. In words, bisimilarity between two worlds implies modal equivalence (true formulas).
\begin{theorem}\label{thm:modalequivalence} Given two Kripke pointed models $({\cal M},w)$ and $(\mathcal{M'},w')$. If ${\cal M}, w\, \underline{\leftrightarrow}\, {\cal M'}, w'$, then $({\cal M},w)$ and $(\mathcal{M'},w')$ satisfy exactly the same $\lanset_{\square}$-formulas.
\end{theorem}

\section{A simple modal logic for bisimulations}\label{sec:logic}

\subsection{Syntax}

\begin{definition}\label{def:lb}
    
The language $\lb$ is the extension of $\lanset_{\square}$ generated by the following grammar:
\begin{center}
$A::=  p \mid \perp \mid (A \to A) \mid \square A \mid \nbi A $
\end{center}

\noindent where $p$ ranges over $\atm$. %

\end{definition}

 The meaning of $\nbi A$ is that ``$A$ is true in all states corresponding to the current one''.
The dual of $\nbi$ is $\bi$, so that $\bi A$  reads ``$A$ is true at least one corresponding state of the current state'', and it is defined $\neg\nbi\neg A$.
 The underlying idea is that  $\nbi$ (and its dual $\bi $) refer in the models only to the points of the domain of a {\em correspondence} relation (as a particular case, the correspondence relation may be a bisimulation).

\subsection{Semantics}

\begin{definition}[Frames for $\lb$]
A {\em frame} for  $\lb$ (or simply, a frame) is a tuple ${\mathfrak F} = ({\cal F}, {\cal F'}, Z)$ where:
\begin{itemize}
\item ${\cal F}, {\cal F'}\in \fset$, where ${\cal F} = (W, R)$ and ${\cal F}' = (W', R')$.
\item  $Z\subseteq W\times W'$ ($Z$ is a relation called \emph{correspondence relation}). 
\end{itemize} 
A frame where $Z\neq \emptyset$ and it satisfies properties \textsf{forth} and \textsf{back} of Def.\ \ref{def:bisimulation} is called a {\em bisimulation frame} (or simply a {\em bi-frame}). 
\end{definition}

\begin{definition}[Models for $\lb$]
A {\em model} for  $\lb$ (or simply, a model) is a tuple ${\mathfrak M} = ( {\mathfrak F}, {\cal M}, {\cal M'})$ where:
\begin{itemize}
\item ${\mathfrak F} = ({\cal F}, {\cal F'}, Z)$ is a frame for  $\lb$. 

\item ${\cal M}$ is a Kripke model based on ${\cal F}$ and ${\cal M'} $ is a Kripke model  based on ${\cal F'}$.

\end{itemize} 
We say in this case that the model ${\mathfrak M}$ is based on ${\mathfrak F}$. 
A model based on a bi-frame that moreover satisfies condition 1 of Def.\ \ref{def:bisimulation} is called a \emph{bi-model}. Given a frame ${\mathfrak F} = ({\cal F}, {\cal F'}, Z)$, in what follows we simplify   the model ${\mathfrak M} = ( {\mathfrak F}, {\cal M}, {\cal M'})$   by means of ${\mathfrak M} = ({\cal M}, {\cal M'}, Z)$.
\end{definition}

\begin{definition}[Truth for $\lb$]\label{def:Truth} Let ${\mathfrak M} = ({\cal M}, {\cal M'}, Z)$ be a model for $\lb$. Let $w$ be in ${\cal M}$, and $w'$ be in ${\cal M'}$. Then for any formulas $A$ and $B$, we define the \emph{truth relation}, that relates triples of the form $({\mathfrak M}, {\cal M}^{(\prime)}, w^{(\prime)})$ and formulas,  as follows:

 \begin{itemize}
 \item ${\mathfrak M}, {\cal M}, w \models p$ iff $w \in V(p)$ (for any $p \in \atm$).
  \item ${\mathfrak M}, {\cal M}', w' \models p$ iff $w \in V'(p)$ (for any $p \in \atm$).
 \item ${\mathfrak M}, {\cal M}, w \not \models \perp$.
  \item ${\mathfrak M}, {\cal M}', w' \not \models \perp$.
  \item ${\mathfrak M}, {\cal M}, w \models A \to B$ iff either ${\mathfrak M}, {\cal M}, w \not\models A $ or ${\mathfrak M}, {\cal M}, w \models B$.
\item ${\mathfrak M}, {\cal M}', w' \models A \to B$ iff either ${\mathfrak M}, {\cal M}', w' \not\models A $ or ${\mathfrak M}, {\cal M}', w'\models B$.

\item ${\mathfrak M}, {\cal M}, w \models\square  A$ iff for every $v \in W$, $wRv$ implies ${\mathfrak M},\cmod,  v \models A$.
\item ${\mathfrak M}, {\cal M'}, w' \models\square  A$ iff for every $v' \in W'$, $w'R'v'$ implies ${\mathfrak M},\cmod', v '\models A$.

\item  ${\mathfrak M}, {\cal M}, w \models  [\,b\,] A$  iff  for all  $w'\in W'$  s.t. $wZw'$,   ${\mathfrak M}, \cmod',w'\models   A$.
\item  ${\mathfrak M}, {\cal M}', w' \models  [\,b\,] A$.

 \end{itemize}

When the context is clear, we simplify ${\mathfrak M}, {\cal M}^{(\prime)}, w^{(\prime)} \models A$ as ${\cal M}^{(\prime)}, w^{(\prime)}\models A$ or as ${\mathfrak M}, w^{(\prime)}\models A$ depending on convenience.
A formula $A$ is \emph{valid in a model} ${\mathfrak M}$ iff ${\mathfrak M}, w^\ast\models A$ for all $w^\ast \in W\cup W'$. A formula is \emph{valid in a frame} ${\mathfrak F}$ iff it is valid in every model ${\mathfrak M}$ based on ${\mathfrak F}$. {A formula $A$ is the (local) semantic consequence of a set $\Gamma \subseteq \lb$ (in symbols, $\Gamma \models A$) iff for every pointed model $(\fmod,\cmod^{(\prime)},w^{(\prime)})$, $\fmod,\cmod^{(\prime)},w^{(\prime)} \models B$ for all $B \in \Gamma$ implies $\fmod,\cmod^{(\prime)},w^{(\prime)} \models A$.}
   \end{definition}

\begin{remark}\label{remark:sharedWorlds} Note that the definition of frame does not assume that $W\cap W'=\emptyset$, for the sake of generality. However, even if $w \in W\cap W'$, this object is not considered equally from a modal point of view, since we might have a situation in which $\mathfrak{M},\mathcal{M},w \models \lnot \nbi \perp$, but it is always the case that $\mathfrak{M},\mathcal{M}',w \models \nbi \perp$ 
by the definition of truth. Hence, our definition implicitly assumes that the position of $\cmod^{(')}$ within $\fmod$ (left or right) is relevant when evaluating $\nbi$-formulas. A similar solution is adopted explicitly in our Isabelle encoding.
\end{remark}

\par\medskip
The clauses of satisfaction for the defined connective $ \bi$ are {in consequence}:
\begin{itemize}

  \item   ${\mathfrak M},{\cal M}, w\models \bi A$ iff there exists $w'\in W'$ such that $wZw'$ and ${\mathfrak M}',{\cal M'},  w'\models  A$.  
   \item  ${\mathfrak M}, {\cal M}', w'\not\models \bi A$ .  
\end{itemize}

\begin{ex}[A bi-model] 

The figure below represents a bi-model, $\fmod_1$, based on a bi-frame, $\ffra_1,$ where $R$ and $R'$ are represented through normal arrows and $Z$ is represented through dashed arrows. We exemplify the notion of truth and validity by stating some facts:

\begin{itemize}
    \item $\fmod_1,\cmod,w\models \square \bi \lozenge p$. 
        \item $\fmod_1,\cmod',w'\models \lnot\square \bi \lozenge p$. 
        \item $\fmod_1,\cmod,w\models \lozenge \nbi \square (p\lor q)$.
                        \item $\fmod_1,u \models \bi\top\to \square\bi\top$.
        \item $\fmod_1 \models p \to \nbi p$.
        \item $\ffra_1 \models \nbi\bot \lor\square\bi\nbi\bot $.

\end{itemize}

\begin{center}
\begin{tikzpicture}[modal,world/.append style=
{minimum size=1cm}]

\node[world] (w) [label=below:$w$] {$p$};
\node[world] (v) [below=1.5cm of w,label=below:$v$] {$q$};
\node (pos) [below=0.5cm of w] {};
\node[world] (u) [right=0.5cm of pos,label=below:$u$] {$p$};

\node (posl) [left=0.5cm of w] {};
\node (posb) [below=0.1cm of v] {};

\node[draw=black, fit=(posl) (w) (v) (u) (posb),label=below:$\cmod$](mod1) {};

\draw[->] (w) edge[reflexive left] (w);
\draw[->] (u) edge[reflexive left] (u);
\draw[->] (v) edge[reflexive left] (v);
\draw[<->] (w) edge[bend right] (v);
\draw[<->] (w) edge (u);
\draw[<->] (u) edge (v);


\node[world] (wp) [right=4cm of w,label=below:$w'$] {$p$};
\node (posr) [right=0.5cm of wp] {};

\node[world] (vp) [below=1.5cm of wp,label=below:$v'$] {$q$};
\node (posb) [below=0.1cm of vp] {};
\node (posl) [left=0.1cm of vp] {};

\draw[->] (wp) edge[reflexive right] (wp);
\draw[->] (vp) edge[reflexive right] (vp);
\draw[<->] (wp) edge[bend left] (vp);

\node[draw=black, fit=(wp) (vp) (posb) (posr) (posl),label=below:$\cmod'$](mod2) {};


\node (pos1) [below=1cm of v] {};
\node[draw=black, fit=(mod1) (mod2) (pos1),label=below:$\fmod_1$](fmod) {};
\draw[->] (w) edge[dashed] (wp);
\draw[->] (u) edge[dashed] (wp);
\draw[->] (v) edge[dashed] (vp);

\end{tikzpicture}
\end{center}
\end{ex}
\black


\subsection{Some basic model theory}

\begin{definition}[Relative expressivity]
    Given two languages $\lanset$ and $\lanset'$ that are interpreted in the same class of models $\mathfrak{C}$, we say that $\lanset$ is at least as expressive as $\lanset'$ (with respect to $\mathfrak{C}$), abbreviated as $\lanset \succeq_{(\mathfrak{C})}\lanset'$ iff for each $A \in \lanset'$, there is a $B  \in \lanset$ such that $A \equiv B $. We use $\lanset \succ\lanset'$ as an abbreviation of $\lanset \succeq \lanset'$ and $\lanset' \not \succeq \lanset$. 
\end{definition}

\begin{proposition}  $\lb \succ\lanset_{\square}$.
\end{proposition}
\begin{proof}
    The proof follows a standard argument. It is analogous to showing that a bi-modal language is more expressive than a mono-modal language over the class of all bi-relational frames. 
\end{proof}

\begin{definition}[Definability]\label{def:definability} Given a class of frames {(models)}  $\mathfrak{C}$ and a language $\lanset$ interpreted over $\mathfrak{C}$, we say that a property P is $\lanset$-definable (within $\mathfrak{C}$) through a set of formulas $\Gamma$ iff for all $\mathfrak{S} \in \mathfrak{C}$, $\mathfrak{S}$ satisfies P iff $\mathfrak{S}\models A$ for all $A \in \Gamma$. Moreover, we say that a property P is $\lanset$-definable (within $\mathfrak{C}$) iff it is 
    definable with some set of formulas $\Gamma$.
\end{definition}

The following two propositions show that the new modality $\nbi$ is sufficient to talk about the main definitional conditions of bisimulations. 
\begin{proposition} The property \textsf{atomic harmony} (see Def.\ \ref{def:bisimulation}) is $\lb$-definable in the class of all models.
\end{proposition}

\begin{proof}

Consider the set $\Gamma_{AH}=\{l\to \nbi l \mid \text{ for any literal } l\}$.
It suffices to prove that all formulas of this set are valid in a model iff this model satisfies the condition of atomic harmony. Details are omitted. 


\end{proof}

\begin{remark}
    It is clear from the previous proposition that the condition \textsf{atomic harmony} for $\Phi$-bisimulations (bisimulations where \textsf{atomic harmony} is restricted to a set of atoms $\Phi\subseteq \atm$) is definable by restricting the axiom scheme $l\to \nbi l$ to literals whose positive parts are in $\Phi$.
\end{remark}

\begin{proposition}\label{prop:definability-forth-back} Properties \textsf{forth} and \textsf{back} (see Def.\ \ref{def:bisimulation}) are $\lb$-definable in the class of all frames.
\end{proposition}
\begin{proof}
    
For \textsf{forth} consider the formula $$(\bi p\land \lozenge [\,b\,] q)\to \bi (p\land \lozenge q)$$All we have to prove is that if ${\mathfrak F} = ({\cal F}, {\cal F'}, Z)$ is a frame for $\lb$, the following condition holds:
\begin{itemize}
\item[] $(\bi p\land \lozenge [\,b\,]q)\to \bi (p\land \lozenge q)$ is valid in ${\mathfrak F}$ iff $Z$ satisfies \textsf{forth}.
\end{itemize}

{
This can be proved by a direct semantic argument.
Details are omitted due to their routine nature.}
For \textsf{back}, the defining formulas $\bi \lozenge p\to \lozenge \bi p$.
Details are again omitted. 
\end{proof}

\begin{remark}
    It follows from the previous proposition that simulations (bisimulations without \textsf{back}) and refinements (bisimulations without \textsf{forth}) are also definable in $\lb$. It is enough to drop the corresponding formula.
\end{remark}

Moreover, we can also show that the inclusion of $\nbi$ in the language is necessary for capturing bisimulations.
\begin{proposition} The property \textsf{atomic harmony} is not $\basiclan$-definable in the class of all models.
\end{proposition}
\begin{proof}
The following two models can be shown to satisfy the same $\basiclan$-formulas world-by-world. This can be proved by an easy induction on the structure of $\basiclan$-formulas. Then, both models validate the same $\basiclan$-formulas. Since $\fmod_1$ satisfies \textsf{atomic harmony} but $\fmod_2$ does not, then the property is not $\basiclan$-definable.

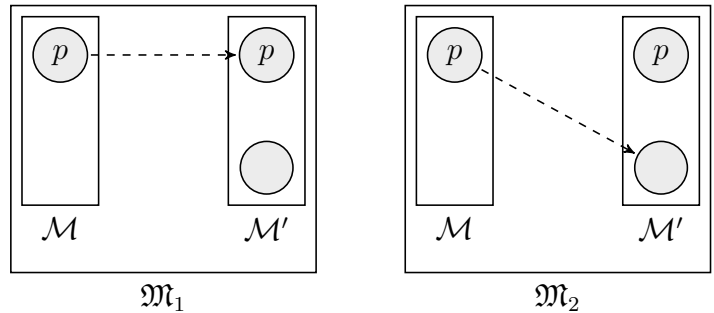
\begin{center}
\begin{tikzpicture}[modal,world/.append style=
{minimum size=0.7cm}]
\node[world] (w) [] {$p$};
\node[tworld] (pos) [below of=w] {};
\node[world] (wp) [right=2cm of w] {$p$};
\node[world] (v) [below of=wp] {};
\draw[->] (w) edge[dashed] (wp);

\node (pos1) [below=0.6cm of pos] {};
\node[draw=black, fit=(w) (pos),label=below:$\cmod$](mod1) {};
\node[draw=black, fit=(wp) (v),label=below:$\cmod'$](mod2) {};
\node[draw=black, fit=(mod1) (mod2) (pos1),label=below:$\fmod_1$](fmod) {};

\node[world] (w) [right=4.5cm of w] {$p$};
\node[tworld] (pos) [below of=w] {};
\node[world] (wp) [right=2cm of w] {$p$};
\node[world] (v) [below of=wp] {};
\draw[->] (w) edge[dashed] (v);

\node (pos1) [below=0.6cm of pos] {};
\node[draw=black, fit=(w) (pos),label=below:$\cmod$](mod1) {};
\node[draw=black, fit=(wp) (v),label=below:$\cmod'$](mod2) {};
\node[draw=black, fit=(mod1) (mod2) (pos1),label=below:$\fmod_2$](fmod) {};

\end{tikzpicture}
\end{center}
    
\end{proof}

\begin{proposition} Properties \textsf{forth} and \textsf{back} are not $\lanset_\square$-definable in the class of all frames.
\end{proposition}

\begin{proof}[Sketch of the proof.] The following two frames can be shown to validate the same $\basiclan$-formulas, because they are isomorphic if we ignore $Z$ (as $\basiclan$ does). The one on the left-hand side satisfies both \textsf{forth} and \textsf{back}, but the one on the right-hand side does not; hence, these properties are not $\basiclan$-definable.

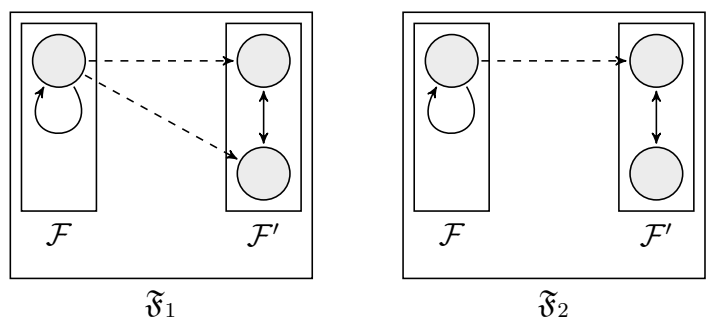
\begin{center}
\begin{tikzpicture}[modal,world/.append style=
{minimum size=0.7cm}]


\node[world] (w) [] {};
\node[tworld] (pos) [below of=w] {};
\draw[->] (w) edge[reflexive below] (w);
\node (pos1) [below=0.6cm of pos] {};
\node[draw=black, fit=(w) (pos),label=below:$\cfra$](mod1) {};


\node[world] (wp) [right=2cm of w] {};
\node[world] (v) [below of=wp] {};
\draw[<->] (wp) edge (v);

\node[draw=black, fit=(wp) (v),label=below:$\cfra'$](mod2) {};


\node[draw=black, fit=(mod1) (mod2) (pos1),label=below:$\ffra_1$](fmod) {};
\draw[->] (w) edge[dashed] (wp);
\draw[->] (w) edge[dashed] (v);


\node[world] (w) [right=4.5cm of w] {};
\node[tworld] (pos) [below of=w] {};
\draw[->] (w) edge[reflexive below] (w);
\node (pos1) [below=0.6cm of pos] {};
\node[draw=black, fit=(w) (pos),label=below:$\cfra$](mod1) {};


\node[world] (wp) [right=2cm of w] {};
\node[world] (v) [below of=wp] {};
\draw[<->] (wp) edge (v);

\node[draw=black, fit=(wp) (v),label=below:$\cfra'$](mod2) {};


\node[draw=black, fit=(mod1) (mod2) (pos1),label=below:$\ffra_2$](fmod) {};
\draw[->] (w) edge[dashed] (wp);

\end{tikzpicture}
\end{center}

\end{proof}

The next proposition highlights an undefinable property of bi-frames that will be shown to be harmless for completeness purposes.

\begin{proposition}\label{Znovacia} The non-emptiness of $Z$ is not $\lb$-definable in the class of all frames.
\end{proposition}

\begin{proof}[Sketch of the proof]
 Consider the following frame:

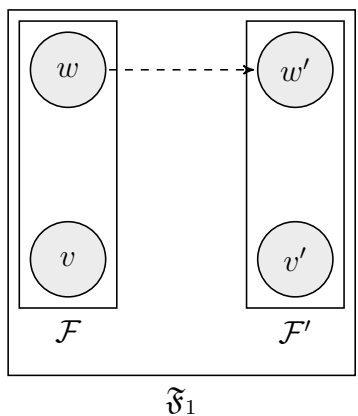
\begin{center}
\scalebox{0.7}{
\begin{tikzpicture}[modal,world/.append style=
{minimum size=1cm}]


\node[world] (w) [] {$w$};
\node[world] (v) [below=1.5cm of w] {$v$};
\node[draw=black, fit=(w) (v),label=below:$\cfra$](mod1) {};


\node[world] (wp) [right=2cm of w] {$w'$};
\node[world] (vp) [below=1.5cm of wp] {$v'$};

\node[draw=black, fit=(wp) (vp),label=below:$\cfra'$](mod2) {};


\node (pos1) [below=0.6cm of v] {};
\node[draw=black, fit=(mod1) (mod2) (pos1),label=below:$\ffra_1$](fmod) {};
\draw[->] (w) edge[dashed] (wp);

\end{tikzpicture}
}
\end{center}

 Note that $\ffra_{\{v,v'\}}=((\{v\},\emptyset),(\{v'\},\emptyset),\emptyset)$ is a generated subframe of $\ffra_1$ (see definitions 2.5 and 3.13 of \cite{blackburn2010modal} for the precise definition). As such, being $\square$ and $\nbi$ normal modal operators, generated submodels can be shown to preserve validity (\cite[Theorem 3.14]{blackburn2010modal}). But then, we have a frame with $Z\neq \emptyset$ ($\ffra_1$) and one with $Z=\emptyset$ ($\ffra_{\{v,v'\}}$) such that all $\lb$-formulas that are valid in the former are also valid in the latter. Hence, the non-emptiness of $Z$ is not $\lb$-definable in the class of all frames.
\end{proof}

\subsection{Axiom system}

We now introduce an axiomatisation in $\lb$ of the valid formulas in the class of all bi-models, called \textsf{L}$_B$, which will later be shown to be sound and complete. The calculus \textsf{L}$_B$ is given by the following axiom schemes and rules:





\begin{center}
    
\begin{tcolorbox}[
    colback=white, 
    colframe=black, 
    arc=0mm,        
    boxrule=0.5pt,  
    left=5pt,       
    right=5pt,      
    width=8cm
]
    
\noindent\textbf{AXIOM SCHEMES}:
\begin{description}
\item[\textbf{(TAU)}] All propositional tautologies
\item[${\bf (K\square)}$] $\square(A\to B)\to (\square A\to \square B)$
\item[${\bf (K\nbi)}$] $\nbi(A\to B)\to (\nbi A\to \nbi B)$
\item[${\bf (FORTH)}$]  $(\bi A\land \lozenge \nbi B)\to \bi (A\land \lozenge B)$
 \item[${\bf (BACK)}$] $\bi \lozenge A\to\lozenge\bi A$ 
\item[${\bf (HARM)}$]  $l\to\nbi l$\quad(where $l$ is a literal)
\item[${\bf (NTS)}$]  $\nbi \nbi \perp$  
\end{description}
\par \medskip
\noindent \textbf{RULES}:
\begin{description}
\item[\textbf{(MP)}] If $\vdash A\to B$ and $\vdash A$, then $\vdash B$.\par

\item[${\bf (N\square)}$]  If $\vdash A$, then $\vdash \square A$

\item[${\bf (N\nbi)}$]  If $\vdash A$, then $\vdash \nbi A$ 
\end{description}

\end{tcolorbox}

\end{center}

The notions of \textsf{L}$_B$-proof, \textsf{L}$_B$-theoremhood (shortened $\vdash A$), \textsf{L}$_B$-de\-du\-ci\-bi\-lity from assumptions (shortened $\Gamma \vdash A$), and \textsf{L}$_B$-consistency are standard \cite{blackburn2010modal}. \par \medskip

The following theorem schema of $\mathsf{L}_B$ captures, in our object language, a fundamental property of bisimulations: bisimilar states satisfy the same modal formulas (see Theorem \ref{thm:modalequivalence}). This is a neat example of how meta-theory about modal logic can now be performed at the object level.

\begin{proposition}\label{resultado} $\vdash A \to \nbi A$ for all $A \in \basiclan$
\end{proposition}






\begin{proof}[Proof sketch]
By structural induction on $A$. The key insight is that it is necessary to simultaneously prove the property for the negated formula, ensuring that the negative version is available within the induction hypothesis. 
Full details can be found in our Isabelle/HOL encoding \cite{Bisimulation_Logic-AFP}.
\end{proof}
We can also formulate and prove the full version of Theorem \ref{thm:equisat} in \textsf{L}$_B$.
\begin{proposition}\label{resultado2} $\vdash \bi\top\to( A \leftrightarrow \nbi A)$ for all $A \in \basiclan$.
\end{proposition}

\begin{proof}\par   
1.\quad $A\to \nbi A$ \hfill Prop.\ref{resultado} \\ 
2.\quad $\bi\top\to(A \to\nbi A)$ \hfill from 1 by propositional logic \\ 
3.\quad $\neg A\to\nbi \neg A$ \hfill Prop.~\ref{resultado} \\
4.\quad $(\nbi A\land \neg A) \to (\nbi A\land \nbi\neg A)$ \hfill from 3 by propositional logic\\
5.\quad$(\nbi A\land \nbi \neg A) \to \nbi\bot$ \hfill Theorem of $K$\\ 
6.\quad $(\nbi A\land \neg A) \to \nbi\bot$ \hfill from 4, 5 by propositional logic \\ 
7.\quad $\bi\top\to (\nbi A\to A)$ \hfill from 6 by  propositional logic and definitions of\\ \indent\qquad\qquad \qquad\qquad\qquad \qquad\qquad\qquad\qquad\qquad\qquad\qquad\qquad   $\top$ and $\bi$ \\
8.\quad $\bi\top\to(A\leftrightarrow \nbi A )$ \hfill from 2, 7 by propositional logic

\medskip 
\end{proof}





\section{Soundness and completeness}\label{sec:completeness}

This section is devoted to the proof of soundness and completeness of $\mathsf{L}_B$ with respect to the class of all bi-models.
{Soundness of \textsf{L}$_B$ with respect to the class of all bi-models} is straightforward. For completeness, we employ the {\em canonical model} technique. All our proofs have been verified in Isabelle/HOL \cite{Bisimulation_Logic-AFP}, hence most of the details are omitted here.
\par \medskip

We call $MC$ the set of all maximal \textsf{L}$_B$-consistent sets (mc-sets, for short).  In what follows, familiarity with the basic properties of mc-sets is assumed.
 
\begin{lemma}[Lindenbaum's Lemma]\label{lindenbaum}
 For all consistent sets $\Delta$, there exists an mc-set $\Gamma$ that contains $\Delta$.
\end{lemma}

\begin{definition}[Canonical relations]
\label{def:CanRel}
We define the relations ${\cal R},{\cal Z}\subseteq MC\times MC$  as follows:

\begin{enumerate}
\item $\Gamma  {\cal R}\Delta$ iff $\{A\mid \square A \in\Gamma\}\subseteq\Delta$.
\item $\Gamma{\cal Z}\Delta$ iff $\{A\mid \nbi A\in\Gamma \}\subseteq\Delta$.
 \end{enumerate}
\end{definition}

The following proposition can be proved by standard means. It is used to establish $\Gamma {\cal R}\Delta$ and $\Gamma{\cal Z}\Delta$ based on $\lozenge$ and $\bi$, respectively:

\begin{proposition}\label{prop:revRZ}Given two mc-sets $\Gamma$ and $\Delta$,
\begin{enumerate}
\item $\Gamma {\cal R}\Delta$ iff $\{\lozenge A\mid A \in\Delta\}\subseteq\Gamma$.
\item $\Gamma{\cal Z}\Delta$ iff $\{\bi A\mid A\in\Delta \}\subseteq\Gamma$.
 \end{enumerate}    
\end{proposition}

\begin{definition}[Canonical domains]\label{def:CanDomains}
    We define the following sets of mc-sets:
    \begin{itemize}
        \item $MC_1  =  MC$.
        \item $MC_2  =  \{ \Gamma \in MC \mid \forall n \in \mathbb{N}, \, \square^n \nbi \perp \in \Gamma\}$.
    \end{itemize}
\end{definition}

\begin{convention}\label{conve:canonicalSide}
Since $MC_1 \cap MC_2 \neq \emptyset$, but we want this to hold for convenience, we assume that elements in $MC_1 \cap MC_2$ are tagged when packed into $MC_2$, so that their occurrences in each of the two sets ($MC_1$ and $MC_2$) can be distinguished. See Remark \ref{remark:sharedWorlds} for the reason behind this choice.
\end{convention}

\begin{definition} [Canonical model]\label{CanMod}
The  {\em canonical model} for \textsf{L}$_B$ is the tuple ${\mathfrak M}_c = ( {\cal M}_1, {\cal M}_2, \cal{Z})$ where:
\begin{itemize}
\item ${\cal M}_1 = ({\cal F}_1, {\cal V}_1)$, where:
\begin{itemize}

\item  ${\cal F}_1 = (MC_1, \mathcal{R}_1)$, where ${\cal R}_1 = \mathcal{R}$.
\item  $\mathcal{V}_1 (p) = \{\Gamma\in MC_1\mid p\in \Gamma\}$, for any $p\in \atm$.
\end{itemize}
\item ${\cal M}_2 = ({\cal F}_2, {\cal V}_2)$, where:
\begin{itemize}

\item  ${\cal F}_2= (MC_2, {\cal R}_2)$, where ${\cal R}_2= {\cal R} \cap (MC_2\times MC_2)$.

\item  $\mathcal{V}_2 (p) = \{\Gamma\in MC_2\mid p\in \Gamma\}$, for any $p\in \atm$.

\end{itemize}

\end{itemize} 

\end{definition}

\begin{lemma}[Existence]\label{lemma:existence} We have the following:
\begin{enumerate}
\item\label{exI1}  If $\lozenge A\in \Gamma$ and $\Gamma \in MC_1$, then there exists $\Delta \in MC_1$ such that $\Gamma{\cal R}_1\Delta$ and $A\in\Delta$.
\item\label{exI2}  If $\lozenge A\in \Gamma$ and $\Gamma \in MC_2$, then there exists $\Delta \in MC_2$ such that $\Gamma{\cal R}_2\Delta$ and $A\in\Delta$.
\item\label{exI3} If $\bi A\in \Gamma$ and $\Gamma \in MC_1$, then there exists $\Delta \in MC_2$ such that $\Gamma{\cal Z}\Delta$ and $A\in\Delta$.
\end{enumerate}
\end{lemma}

\begin{proof}
    The proof of existence lemmas is common in modal logic. In our case, we also need to establish the following results:
    \begin{eqnarray}
    \label{eq:exP2}
    & \{A  \mid \square A  \in \Gamma\}\subseteq \Lambda \text{ and } \Gamma \in MC_2 \quad  \text{implies} \quad  \Lambda \in MC_2 &\\
    \label{eq:exP3}
    & \{A  \mid \nbi A  \in \Gamma\}\subseteq \Lambda  \quad  \text{implies} \quad  \Gamma \notin MC_2\ \ \text{and}\ \ \Lambda \in MC_2 &
    \end{eqnarray}
    
See \cite{Bisimulation_Logic-AFP} for the full verified proof.


\end{proof}

\begin{lemma}[Canonicity]\label{lemma:Cannonicity} The canonical model for \textsf{L}$_B$ is a bi-model (for $\lanset_B$). More detailed:

\begin{enumerate}
    \item\label{canonNonempty} $\mathcal{Z}\neq \emptyset$.
    \item\label{canonZ} $\mathcal{Z}\subseteq MC_1 \times MC_2$.
    \item\label{canonAH} $\mathfrak{M}_c$ satisfies \textsf{atomic harmony}.
    \item\label{canonForth} $\mathfrak{M}_c$ satisfies \textsf{forth}.
    \item\label{canonBack} $\mathfrak{M}_c$ satisfies \textsf{back}.
\end{enumerate}
\end{lemma}

\begin{proof} 
We illustrate the argument by proving item \ref{canonNonempty}; the rest can be found in \cite{Bisimulation_Logic-AFP}.
We first show that $\{\bi \top\}$ is consistent. For suppose it were not. Then $\vdash_B\lnot \bi \top$, which implies, by soundness of $\mathsf{L}_B$, that $\lnot \bi \top$ is valid in all bi-models. Note that this is not the case (take any model with $Z\neq \emptyset$ and any $w \in W$ such that $wZw'$). Now, since  $\{\bi \top\}$ is consistent, by the Lindenbaum's Lemma we have that $\{\bi \top\}\subseteq \Gamma$ for some maximally consistent $\Gamma$. But then, by Lemma~\ref{lemma:existence}, we have that there is a $\Delta \in MC_2$ such that $\Gamma \mathcal{Z}\Delta$.
\end{proof}

\begin{lemma}[Truth]
\label{truthLemma}
For all mc-set $\Gamma$ and every $A \in L_{\square\nbi}$,
\begin{itemize}
    \item  If $\Gamma \in MC_1$, then $A \in \Gamma$ iff $\fmod_c,\cmod_1,\Gamma\models A$.
\item  If $\Gamma \in MC_2$, then $A \in \Gamma$ iff $\fmod_c,\cmod_2,\Gamma\models A$.
    
\end{itemize}
\end{lemma}

\begin{proof}
    The proof is done by induction on $A$. Details are omitted but can be found in \cite{Bisimulation_Logic-AFP}.

\end{proof}

\begin{theorem}[Strong completeness] For every $\Gamma \subseteq \lb$ and every $A\in \lb$, $\Gamma \models A$ implies $\Gamma\vdash A$.
\end{theorem}

\begin{proof}
    The proof is standard by using Lemmas~\ref{lindenbaum} and~\ref{truthLemma}, and properties of $\vdash$ and $\models$ (see, e.g., \cite{blackburn2010modal}).
\end{proof}

\section{Decidability and complexity}\label{sec:decidability}
This section is devoted to showing that our logic is decidable and, moreover, PSPACE-complete via a translation to standard modal logic $K$. \par 
  We extend our vocabulary with a fresh propositional variable $E_b \notin \atm$,
which will intuitively stand for ``there exists a bisimilar
world''.  We define two translation functions, $\tau_1$ and $\tau_2$,
which map formulas from $\mathcal{L}_{\Box[b]}$ to the basic modal
language $\basiclan$ over $\atm \cup \{E_b\}$. 

\par 

The translation $\tau_2 : \mathcal{L}_{\Box[b]} \to \basiclan$
  evaluates formulas strictly in the right-hand model $\cmod_2$, where no outgoing bisimulation relations exist. In this vein, $\tau_1 : \mathcal{L}_{\Box[b]} \to \basiclan$
  evaluates formulas in the left-hand model $\cmod_1$, dynamically
  checking for the existence of a bisimilar world via $E_b$.
\begin{definition}[Translation Functions $\tau_1$ and $\tau_2$]
  
  \begin{align*}
    \tau_2(p) &= p \quad \text{for } p \in \atm \\
    \tau_2(\neg  A) &= \neg \tau_2( A) \\
    \tau_2( A \land  B) &= \tau_2( A) \land \tau_2( B) \\
    \tau_2(\Box  A) &= \Box \tau_2( A) \\
    \tau_2(\nbi  A) &= \top
  \end{align*}
  \begin{align*}
    \tau_1(p) &= p \quad \text{for } p \in \atm \\
    \tau_1(\neg  A) &= \neg \tau_1( A) \\
    \tau_1( A \land  B) &= \tau_1( A) \land \tau_1( B) \\
    \tau_1(\Box  A) &= \Box \tau_1( A) \\
    \tau_1(\nbi  A) &= E_b \to \tau_2( A)
  \end{align*}
\end{definition}

Note that $\tau_1(\langle b \rangle  A) = \lnot(E_b\to\lnot \tau_2( A))$ which is equivalent to $ E_b \land \tau_2( A) $. It is also immediate to show that $\tau_2(\langle b \rangle  A)$ is equivalent to $\bot$. \\

To simulate the structural requirement of the \textsf{forth}
condition, we define the global axiom $A_b$:
$A_b := E_b \to \Box E_b$.
{Intuitively, if a world has a bisimilar counterpart, then every world accessible from it via $R$ also has one. }


\begin{theorem}[Equisatisfiability]\label{thm:equisat}
  A formula $ A \in \mathcal{L}_{\Box[b]}$ is satisfiable in a
  bi-model $\fmod  = (\cmod_1, \cmod_2, Z)$ if and only if its
  translation $\tau_1( A)$ is satisfiable in a standard Kripke model
  $\cmod = (W, R, V)$ where the global axiom $A_b$ holds in every world. 
\end{theorem}
\begin{proof}
   $(\Rightarrow)$ Assume $A$ is satisfiable in a bi-model. Without loss of generality, we can assume it is satisfied in the left-hand model, i.e., there is a bi-model $\mathfrak{M} = (\mathcal{M}_1, \mathcal{M}_2, Z)$ where $\mathcal{M}_1 = (W_1, R_1, V_1)$ and $\mathcal{M}_2 = (W_2, R_2, V_2)$, such that $\mathfrak{M}, \mathcal{M}_1, w \models A$ for some world $w \in W_1$. (If $A$ were instead satisfied in the right-hand model $\mathcal{M}_2$ at some world $u$, we could construct a new bi-model where $\mathcal{M}_2$ acts as the left-hand model, and $Z$ consists solely of a link between two fresh, isolated dummy worlds added to each side with identical valuations. In this new bi-model, $u$ is on the left side but has no $Z$-successors; thus, all $[b]$-formulas remain vacuously true at $u$, perfectly mimicking its original evaluation in the right-hand model). We construct the standard Kripke model $M = (W, R, V)$ exactly as follows:
  \begin{itemize}
  \item $W = W_1$
  \item $R = R_1$
  \item $V(p) = V_1(p)$ for all $p \in \atm$
  \item $V(E_b) = \{ x \in W_1 \mid \exists u \in W_2 \text{ such that } (x, u) \in Z \}$
  \end{itemize}

  First, we show that $M$ globally satisfies the axiom $A_b := E_b \to \Box E_b$. Assume that for an arbitrary world $x \in W$, we have $M, x \models E_b$. By definition of $V(E_b)$, there exists a world $u \in W_2$ such that $xZu$. Let $y \in W$ be any world such that $xRy$. Since $R = R_1$, we have $xR_1y$. By the \textsf{forth} condition of the bi-model (Def.~\ref{def:bisimulation}), there must exist $v \in W_2$ such that $uR_2v$ and $yZv$. Because $yZv$ holds, it follows by definition that $y \in V(E_b)$, meaning $M, y \models E_b$. Since this holds for every $R$-successor $y$ of $x$, we conclude $M, x \models \Box E_b$. Thus $M \models A_b$.

  Next, we show that $\mathfrak{M}, \mathcal{M}_1, w \models A$ if and only if $M, w \models \tau_1(A)$ by structural induction on $A$. The propositional, boolean, and $\Box$ cases follow directly from the induction hypothesis and the fact that $M$ shares the same relational structure and basic valuation as $\mathcal{M}_1$.

  The crucial case is $A = [b]B$. We must show that $\mathfrak{M}, \mathcal{M}_1, w \models [b]B$ iff $M, w \models E_b \to \tau_2(B)$. Notice two facts regarding $\tau_2$:
  \begin{enumerate}
      \item Since $[b]$ evaluates vacuously to true in the right-hand model $\mathcal{M}_2$, evaluating $B$ in $\mathcal{M}_2$ is strictly equivalent to evaluating its translation: for any $u \in W_2$, $\mathfrak{M}, \mathcal{M}_2, u \models B \iff \mathfrak{M}, \mathcal{M}_2, u \models \tau_2(B)$.
      \item By definition, $\tau_2(B)$ is a formula in the basic modal language $\basiclan$ (all $[b]$ operators have been eliminated). If $w \in W_1$ and $u \in W_2$ are linked by $wZu$, they are bisimilar. By Theorem \ref{thm:modalequivalence}, they satisfy the exact same $\basiclan$-formulas. Therefore, $\mathfrak{M}, \mathcal{M}_2, u \models \tau_2(B) \iff M, w \models \tau_2(B)$ (evaluating in $M$ is equivalent to $\mathcal{M}_1$ since domains and relations are identical).
  \end{enumerate}

  Using these two facts, the inductive step unfolds naturally without requiring a simultaneous induction over $\tau_2$:
  \begin{align*}
      \mathfrak{M}, \mathcal{M}_1, w \models [b]B &\iff \text{for all } u \in W_2 \text{ such that } wZu, \mathfrak{M}, \mathcal{M}_2, u \models B \\
      &\iff \text{for all } u \in W_2 \text{ such that } wZu, M, w \models \tau_2(B)\\ &\phantom{\iff} \text{ (by facts 1 and 2)} \\
      &\iff \text{if } (\exists u \in W_2 \text{ such that } wZu) \text{ then } M, w \models \tau_2(B) \\
      &\iff \text{if } M, w \models E_b \text{ then } M, w \models \tau_2(B) \text{ (def. of } V(E_b)) \\
      &\iff M, w \models E_b \to \tau_2(B) \\
      &\iff M, w \models \tau_1([b]B)
  \end{align*}
 
\noindent{}$(\Leftarrow)$ Let $M = (W, R, V)$ be a Kripke model such that
  $M, w \models \tau_1( A)$ and 
  $M, x \models E_b \to \Box E_b$ for all $x \in W$. 

  We construct a bi-model $\fmod = (\cmod_1, \cmod_2, Z)$ from $M$. To guarantee that the domain $W_2$ in $\cmod_2 = (W_2, R_2, V_2)$ is non-empty even if no world in $M$ satisfies $E_b$, we introduce a fresh, isolated dummy world $d \notin W$. The definition of $\fmod$ is:
  \begin{itemize}
  \item $\cmod_1 = (W \cup \{d\}, R, V \restriction \atm)$ 
  \item $W_2 = \{ x \in W \mid M, x \models E_b \} \cup \{d\}$
  \item $R_2 = R \cap (W_2 \times W_2)$
  \item $V_2(p) = V(p) \cap W_2$ for all $p \in \atm$
  \item $Z = \{ (x, x) \mid x \in W_2 \}$
  \end{itemize}
  
  We must prove that $Z$ is a valid bisimulation:
  \begin{enumerate}
  \item \textbf{\textsf{atomic harmony}:} If $(x, u) \in Z$, then $x = u$. By
    definition, $V_1$ and $V_2$ agree on $W_2$, so $x$ and $u$ satisfy
    the same propositional variables in $\atm$. 
  \item \textbf{\textsf{forth}:} Suppose $(x, u) \in Z$ and
    $(x, y) \in R$. By definition of $Z$, $x = u$ and $x \in W_2$.
    Since the dummy world $d$ is isolated (it has no $R$-successors), $x$ must be an original world from $W$, meaning $M, x \models E_b$. By the global axiom $A_b$, since
    $(x, y) \in R$, it follows that $M, y \models E_b$. Thus,
    $y \in W_2$. Let $v = y$. We have $(u, v) = (x, y) \in R$, and
    since both are in $W_2$, $(u, v) \in R_2$. Finally,
    $(y, v) = (y, y) \in Z$. 
  \item \textbf{\textsf{back}:} Suppose $(x, u) \in Z$ and
    $(u, v) \in R_2$. By definition, $x = u$ and both $u, v \in
    W_2$. Since $R_2 \subseteq R$, we trivially have
    $(u, v) \in R$. Let $y = v$. Then $(x, y) \in R$ and
    $(y, v) = (v, v) \in Z$. 
  \end{enumerate}
  Since $d$ is topologically isolated from $w$, its addition does not alter the truth value of $\tau_1( A)$ at $w$. Given that $Z$ is a non-empty valid bisimulation and $\cmod_1, w \models \tau_1( A)$, a standard structural induction shows that $\fmod$ satisfies the original formula $A$ at $w \in W_1$. 
\end{proof}

The reduction presented in
Theorem \ref{thm:equisat} effectively maps the satisfiability problem
of $\mathsf{L}_B$ to the problem of global satisfiability
under a restricted frame condition in the standard modal logic
$K$. 

\begin{corollary}The satisfiability problem for $\mathsf{L}_B$ is decidable and is PSPACE-complete. 
\end{corollary}
 
\begin{proof} It is a foundational
  result by Ladner \cite{ladner1977computational} that the satisfiability problem for the
  basic modal logic $K$ is PSPACE-complete.  Generally, evaluating a
  formula subject to an arbitrary global axiom (or equivalently,
  extending the logic with the Universal Modality $U$) increases the
  computational complexity to EXPTIME-complete \cite{spaan1993}.  This
  exponential blow-up occurs because arbitrary global axioms can act
  as ``world generators'' (e.g., $p \to \Diamond p$), creating infinite
  trees that require complex blocking mechanisms to ensure
  termination. However, our conservative reduction only requires the
  global enforcement of the specific condition
  $A_b := E_b \to \Box E_b$.  Crucially, this axiom contains no
  existential modalities ($\Diamond$).  It merely propagates a
  propositional constraint ($E_b$) to successor worlds that are
  \textit{already required to exist} by the existential modalities
  present in the original formula $ A$.  Consequently, $A_b$ cannot
  force the generation of new worlds beyond the modal depth of
  $\tau_1( A)$. Because the tree-model depth remains strictly bounded
  by the modal depth of the original formula, satisfiability can be
  decided using a standard depth-first search (DFS) tableau algorithm,
  which only needs to keep a single branch in memory at any given
  time, thus operating in polynomial space.  Since $\mathsf{L}_B$ conservatively extends $K$, the problem is
  PSPACE-hard.  Therefore, the satisfiability of $\mathsf{L}_B$
  is exactly
  PSPACE-complete. \end{proof}

\section{A logic for $n$-bisimulations}\label{sec:nbisimulations}

The notion of $n$-bisimulation is the standard finite approximation of bisimulations, designed to deal with finite models (see \cite[Chapter 2.3]{blackburn2010modal}). This section is devoted to presenting a logic for $n$-bisimulations that relies heavily on our previous results. Let us start with their formal definition.

\begin{definition}[$n$-Bisimulation for $\lanset_{\square}$]
Let ${\cal M} = (W, R, V)$ and ${\cal M'} = (W', R', V')$ be a pair of Kripke models. For any $n \in \mathbb{N}$, we say that a sequence of non-empty binary relations over $ W \times W'$, $Z_n \subseteq Z_{n-1} \subseteq \dots \subseteq Z_0$ is an {\em $n$-bisimulation between} ${\cal M}$ and ${\cal M'}$ (in symbols: $(Z_k)_{k=0}^n \colon {\cal M} \,\underline{\leftrightarrow}_n\, {\cal M'}$) if for every $k$ such that $0 \le k \le n$, the following conditions hold:
\begin{enumerate}
\item[(1)] If $wZ_0 w'$, then $w$ and $w'$ satisfy the same atoms {\em \textsf{atomic harmony}}.
\item[(2)] If $k > 0$, $wZ_kw'$, and $wRv$, then there exists $v'$ in ${\cal M'}$ such that $vZ_{k-1}v'$ and $w'R'v'$ \emph{(k-forth)}.
\item[(3)] If $k > 0$, $wZ_kw'$, and $w'R'v'$, then there exists $v$ in ${\cal M}$ such that $vZ_{k-1}v'$ and $wRv$ \emph{(k-back)}.
\end{enumerate}
If $w$ (in ${\cal M}$) and $w'$ (in ${\cal M'}$) are a pair of states linked by the $n$-th layer of an $n$-bisimulation $Z_n$, then we say that $w$ and $w'$ are {\em $n$-bisimilar} (written $Z_n\colon {\cal M}, w\, \underline{\leftrightarrow}_n\, {\cal M'}, w'$ or if the context is clear, simply, $w\,\underline{\leftrightarrow}_n\, w'$).

\end{definition}

The fundamental result about $n$-bisimulations is that they preserve modal satisfiability up to modal depth $n$. 

\begin{proposition}\label{thm:n-modalequivalence} Let $({\cal M},w)$ and $(\mathcal{M'},w')$ be two pointed Kripke models. If ${\cal M}, w\, \underline{\leftrightarrow}_n\, {\cal M'}, w'$, then $({\cal M},w)$ and $(\mathcal{M'},w')$ satisfy the exact same $\lanset_{\square}$-formulas of at most modal depth $n$.
\end{proposition}

\paragraph{Syntax for $n$-bisimulations.} Given $n \in \mathbb{N}$, we define the language $\lanset_{B_n}$ as the extension of $\basiclan$ with new modalities $\nbi_k$ for each $0\leq k \leq n$.  

\paragraph{Semantics for $n$-bisimulations.} 
An $n$-model is a tuple $\mathfrak{M} = (\mathcal{M},  \mathcal{M}', Z_0, \dots, Z_{n})$ where $\cal M$ and $\cal M'$ are Kripke models and each $Z_k$ (for $0\leq k \leq n$) is a relation between the domains of $\mathcal{M}$ and $\mathcal{M}'$. We say that $\mathfrak{M}$ is a $n$-bi-model iff $Z_0\dots Z_n$ is an $n$-bisimulation between ${\cal M}$ and ${\cal M}'$. Frames are defined by abstracting away from propositional valuations. As expected, each modality $\nbi_k$ is interpreted in a $n$-model using the relation $Z_k$.

\paragraph{Definability of $n$-bisimulations.} For each $0<k\leq n$, consider the schemas:

\begin{description}
    \item[(Inc)] $\nbi_{k-1}A\to\nbi_{k}A$. 
\item[(${\bf FORTH}_k$)]  $(\bi_k A\land \lozenge [\,b\,]_{k-1} B)\to \bi_{k} (A\land \lozenge B)$.
 \item[(${\bf BACK}_k$)] $\bi_k \lozenge A\to\lozenge\bi_{k-1} A$.
\item[(${\bf HARM}_k$)]  $l\to\nbi_{k} l$\quad(where $l$ is a literal)
\end{description}

\begin{proposition}[Definability of $n$-bisimulations]
    Let ${\frak M}=({\cal M},{\cal M}', Z_0 \dots Z_n)$ be an $n$-frame, we have that:

    $${\frak M}\models {\bf (Inc)} \land {\bf (FORTH_k)}\land ({\bf BACK}_k) \land {\bf (HARM_k)}$$

    \centering iff

    ${\frak M}$ is an $n$-bi-model.
\end{proposition}

\black
\begin{proof}
The structural conditions (inclusion, $k$-forth, and $k$-back) are frame properties captured by the corresponding axioms. For each condition, the corresponding formula is valid in every model whose underlying frame satisfies the condition, and invalid in any model where the frame fails it (by constructing a suitable countermodel). This follows the same direct semantic argument as in the case $n=1$ (Proposition \ref{prop:definability-forth-back}). For models based on frames satisfying the structural axioms, the axiom scheme (HARM) is valid iff the model satisfies atomic harmony.
\end{proof}

\paragraph{Adapting the completeness proof.}
For a sound and complete axiom system with respect to all $n$-bi-models, we must add, to the defining schemas for $n$-bisimulations, all propositional axioms and rules, $K$-axioms and rules for $\square$ and the following axioms and rules (for $0\leq k\leq n$):

\begin{description}
    \item[(K$\nbi_k$)] $\nbi_k(A\to B)\to (\nbi_k A \to \nbi_k B) $ .
        \item[(NTS$_k$)] $\nbi_k\nbi_k \perp $.

    \item[(K$\nbi_0$)] if $\vdash A$, then $\vdash \nbi_0 A$.  
\end{description}

The proof follows the same strategy as in Section \ref{sec:completeness}. Details are omitted.
\section{Related work}\label{sec:discussion}

Several lines of work in modal logic are related to the current paper. Let us comment on them and explain how they differ from our approach. \par 

In his doctoral thesis \cite{hollenberg1998thesis}, Marco Hollenberg introduced \textit{bisimulation quantified modal logics} (see also Tim French's thesis \cite{french2006bisimulation} and related publications \cite{french2005bisimulation,french2006bisimulationAd}). These are standard modal logics extended with propositional quantifiers $\forall x$ and $\exists x$ (where $x$ is an atomic proposition). Such operators are interpreted in standard Kripke models, using a clause that refers to bisimilar models. Bisimulations are here relativised to a set of atoms $\Phi\subseteq \atm$. 
Then, propositional quantifiers are interpreted as follows:
${\cal M},w\models \exists x A$ iff there is some model ${\cal M}',w'$, $\{x\}-$bisimilar to $\cmod,w$, such that ${\cal M}',w'\models A$
 (and something analogous for $\forall x\varphi$). The resulting logics have been studied in detail, showing that undecidability arises frequently (depending on the base modal logic), but that some other desirable properties, such as bisimulation invariance, are conserved. The conceptual and technical differences between this and our approach are significant. First, although both operators quantify over bisimilar states, we do so within a given (pair of) model(s), while in the bisimulation quantifiers approach, this is in the space of all bisimilar states. Second, and most importantly, this approach assumes the notion of bisimulation when interpreting the new operators, while our main research question was the definability of bisimulations via frame correspondence itself.

\par \smallskip

Closer to our work, there exists the long-standing question of how to define relational models up to bisimulation (see, for instance, \cite{fernandez2011modal,van2014definability} for a couple of specific contributions, and \cite{van2011logical} for a textbook presentation of this topic; this theme can be seen as the modal counterpart of \textit{Scott sentences} that define first-order structures up to potential isomorphism). As a sample of this line of work, \cite{van2014definability} faces the inquiry into finding modal formulas in the epistemic multi-agent language (and its extension with the common knowledge operator) that define epistemic models up to bisimulation. The results are largely negative, indicating that standard epistemic languages are insufficient to describe epistemic bisimulations. Although the research spirit is shared with the current paper, the differences are evident again. First, we stick to monomodal logic as the basic logic. 
Second, we also do not consider a transitive closure operator (i.e., common knowledge) in our basic language. Third, we do not assume anything about $R$ while \cite{fernandez2011modal} and \cite{van2014definability} focus on transitive and equivalence relations, respectively. Fourth, we use a special modality (i.e., $\nbi$) to describe bisimulations. Finally, and most importantly, we are using different notions of definability. While we use definability via frame validity (Definition \ref{def:definability}), the mentioned literature talks about definability via modal truth. The last two points explain our positive results for the basic mono-modal case.


\par \smallskip

\par 

From a more general perspective, our semantic setting can be related to different works on two-sorted modal logics, these are logics using frames where the set of worlds is partitioned into two disjoint components, just as we did (see, for instance, \cite{venema1999points,hartonas2019lattice,howlader2023two}). Note that the term \textit{two-sorted} (or, more generally, \textit{multi-sorted}) modal logic is ambiguous in the literature. We use it here as explained, but other authors refer to sorting other sets, e.g., the set of propositional variables, as done in hybrid logic. A quick look at this branch of the literature reveals that two-sorted frames have been used with very different purposes in mind: doing projective geometry, where one sort stands for lines and the other one for points \cite{venema1999points}, lattice theory \cite{hartonas2019lattice}, and formal concept and rough sets theory \cite{howlader2023two}. This varied spectrum of applications seems to motivate a more abstract study of such frames, which, to the best of our knowledge, does not exist in the literature yet (see \cite{balbiani2022parametrized} for the seminal study of a more general class of sorted frames).

\par
\black 
\section{Conclusion}\label{sec:conclusion}

\paragraph{Recap.} Let us start our closing with a brief recap. This paper is inserted into a long-standing tradition of approaching mathematical notions through the use of modal languages and semantics. More specifically, it advances a milestone in the modal study of notions that stem from the meta-theory of standard modal logic, an enterprise that we have coined \textit{meta-modal logic}. Our focus has been on a fundamental notion in the model and correspondence theory of modal logic: bisimulations. Our main findings can be summarised as follows. First, the main definitional conditions of bisimulations for the basic modal language (\textit{atomic harmony}, \textit{forth}, and \textit{back}) can be defined in a quite simple extension of such a language that uses one single additional modality, $\nbi$, to describe the bisimulation relation among two arbitrary Kripke models. Moreover, this modality is shown to be not only sufficient, but also necessary to capture the three mentioned conditions on bisimulations. Second, the set of all valid formulas in this extended language can be captured in a sound and complete axiomatisation, based on our previous definitional results. Third, the satisfiability problem of the logic can be shown to be PSPACE-complete. Finally, we apply a robust method for the main results: our proofs are encoded and verified in Isabelle/HOL, which, in addition, helped us spot some errors in the original, handwritten proof trials. 


\paragraph{Discussion.}
A remarkable consequence of Theorem \ref{thm:equisat} and its corollary is the balance they strike between expressive power
and computational complexity.  As established in Section \ref{sec:logic},
$\mathcal{L}_{\Box[b]}$ is strictly more expressive than the standard
modal language $\basiclan$.  It possesses the ability to natively enforce and
verify bisimulation structures across disconnected models, a feature
that typically pushes bi-dimensional interacting logics into
undecidability (e.g.,\ by forcing an $\mathbb{N} \times \mathbb{N}$
grid-like structures) or at least into EXPTIME. However, because the
$[b]$ operator in our semantics cannot be infinitely chained to
transition into arbitrary new domains, the logic avoids the grid
trap.  More importantly, our reduction proves that this significant
increase in expressivity comes at absolutely no computational cost.  By
absorbing the bisimulation into the well-behaved framework of a
Horn-clause-like global axiom, the logic remains in
PSPACE. Furthermore, this reduction provides a straightforward, highly
efficient pathway for practical implementation and automated
reasoning.  Instead of developing bespoke calculi to handle the
bi-dimensional interaction of $\Box$ and $[b]$ (which often struggle
with infinite modal depth regressions and pairwise blocking), one can
leverage standard analytic tableaux for $K$.  Implementing a decision
procedure for $\mathcal{L}_{\Box[b]}$ requires only a trivial
modification to a standard PSPACE tableau algorithm: simply carrying
the $E_b$ flag forward to all generated successors.  This completely
bypasses the need for the complex EXPTIME dynamic blocking techniques
usually required for Description Logics with General Concept
Inclusions \cite{baader2001overview}, yielding a terminating, sound,
and complete decision procedure that is highly amenable to
optimization. 

\paragraph{Open challenges.} 
As for future work, there are many open venues. 
One of them is the generalisation of $\mathsf{L_B}$ from two frames to $n$ (or even $\omega$) frames. A further step in this direction would be to study the computational feasibility of these logics (in case they are decidable).
An additional research perspective consists of the study of extensions of $\lb$ in order to tackle further meta-modal notions. A natural next target would be surjective p-morphisms \cite{blackburn2010modal}, which are known to preserve validity between Kripke frames. A surjective p-morphism is a bisimulation among frames that is moreover a total, surjective function. We know from classical frame correspondence results that partial functionality of $Z$ is definable in $\lb$ (it corresponds to $\bi p \to \nbi p$), but totality and surjectivity need further expressive resources. Yet another related, important notion that could be studied with object modal languages is bisimilarity (bisimulations that are maximal with respect to set-inclusion).

    
    
    

\bibliographystyle{plain}
\bibliography{bisimulation}
 
\end{document}